\documentclass{article}  

\usepackage{latexsym,amsmath,amsfonts,amssymb,textcomp,comment,amsthm}

\def\Rom#1{\uppercase\expandafter{\romannumeral#1}}

\newcommand{\dollar}[0]{\$}

\newcommand{\TSigma}{\widetilde{\Sigma}}
\newcommand{\tildew}{\widetilde{w}}

\newcommand{\mymatrix}[2]{\left( \begin{array}{#1} #2\end{array} \right)}

\newcommand{\myvector}[1]{\mymatrix{c}{#1}}

\newtheorem{proposition}{Proposition}
\newtheorem{corollary}{Corollary}
\newtheorem{remark}{Remark}

\begin{document}

\title{Decision problems on unary probabilistic and quantum automata\thanks{Some parts of this work was done during Hirvensalo's visit to National Laboratory for Scientific Computing (Brazil) in June 2015 supported by CAPES with grant 88881.030338/2013-01 and during Yakary\i lmaz's visit to University of Turku in September 2016.}
}

\author{
Mika Hirvensalo$^{1,}$\thanks{Partially supported by the V\"ais\"al\"a foundation.} ~~~ Abuzer Yakary\i lmaz$^{2,}$\thanks{Partially supported by CAPES with grant 88881.030338/2013-01 and ERC Advanced Grant MQC.} 
\\ 
\small
$^1$Department of Mathematics and Statistics,
\\ 
\small 
University of Turku, FIN-20014, Turku, Finland
\\
\small
$^2$Faculty of Computing, University of Latvia, 
\\ 
\small
Raina bulv. 19, Riga, LV-1586, Latvia
\\
\small
\textit{mikhirve@utu.fi~~~abuzer@lu.lv}
}

\maketitle    
  
\begin{abstract} 
It is well known that the emptiness problem for binary probabilistic automata and so for quantum automata is undecidable. We present the current status of the emptiness problems for unary probabilistic and quantum automata with connections with Skolem's and positivity problems. We also introduce the concept of linear recurrence automata in order to show the connection naturally. Then, we also give possible generalizations of linear recurrence relations and automata on vectors.
\end{abstract}
%
%
\section{Introduction}

Finite automata are theoretical models for real-time computing with a finite memory. They form a cornerstone
of theoretical computer science, introduced in 1940's and 1950's via a series of papers such as \cite{McCulloch}, \cite{Kleene}, \cite{Mealy}, \cite{Moore}, and \cite{RS59}. Stochastic versions (probabilistic finite automata (PFAs)) were introduced in \cite{Rab63}, and their properties were extensively studied in \cite{Paz71}. In 1997, their quantum versions (quantum finite automata (QFAs)) were introduced in \cite{KW97} and \cite{MC00} (see \cite{AY15} for a recent comprehensive survey).

It is known that the emptiness (and related) problems are undecidable for PFAs on binary alphabets and since they are generalizations of PFAs, they are also undecidable for QFAs on binary alphabets \cite{AY15}. On unary alphabets, on the other hand, these problems are still open and they are known to be decidable only for some small automata.

In this paper, we present the equivalence of the emptiness problems for unary automata to Skolem's problems and positivity problems for linear recurrence relations (LRR). For this purpose, we also introduce automata versions of LRRs and some of their generalizations defined over vectors.

In the next section, we provide the necessary background. Then, we introduce the concept of linear recurrence automata and their possible generalizations in Section \ref{sec:LRA}. The basic relations on LRRs and generalized finite automata are given in Section \ref{sec:basic}. The conversions between automata models are listed in Section \ref{sec:conversions}. Finally, we present the status of emptiness problems for the models in Section \ref{sec:status}. 

\section{Backgrounds}

In this section, we present merely the necessary definitions. For a broader introduction, we refer the reader to \cite{Ei74}, \cite{Yu}, \cite{Paz71}, and \cite{AY15}.

Throughout the paper, $ \Sigma $ denotes the input alphabet not including the right end-marker $ \dollar $, and the extension $ \Sigma \cup \{\dollar\}  $ is denoted by $ \TSigma $. Correspondingly, $\tildew$ represents the string $w \dollar$ for any $w\in\Sigma^*$. Unary alphabet is chosen as $ \Sigma = \{a\} $ and the empty string is denoted by $ \varepsilon $. For a given machine $ M $ and string $ w \in \Sigma^* $, $ f_M(w) $ represents the accepting value/probability of $ M $ on $w$. We may also use symbol $ \diamond $ to present any of relations in set $\Diamond= \{ >,\geq,<,\leq,=, \neq \} $.

\subsection{Automata models}

A generalized finite automaton (GFA) $ G $ is a 5-tuple
\[
	G = ( Q,\Sigma, \{ A_\sigma \mid \sigma \in \Sigma \},v_0,f ),
\]
where
\begin{itemize}
	\item $  Q = \{q_1,\ldots,q_n\} $ is the set of states. 
    \item $ A_\sigma \in \mathbb{R}^{n \times n } $ is the transition matrix for symbol $ \sigma $, i.e. $ A[j,i] $ represents the transition value from the $ i $-th state to $j$-th state,
    \item $ v_0 $ is a $ |Q| $-dimensional real column vector called initial state, and,
    \item $f$ is a $ |Q| $-dimensional real row vector called final vector.
\end{itemize}
For a given input $ w \in \Sigma^* $, $ G $ starts its computation in $ v_0 $, reads $ w $ symbol by symbol from left to the right, and, for each symbol, the state vector is updated by multiplying the corresponding transition matrix. That is,
\[
	v_j = A_{w_{j-1}} v_{j-1},
\]
where $ 1 \leq j \leq |w| $. We denote the final state as $ v_f = v_{|w|} $. The accepting value of $ w $ is calculated as
\[
	f_G(w) = fv_f = f A_{w_{|w|}} \cdots A_{w_1}v_0,
\]
which is to say that the automaton $ G $ computes a function $\Sigma^*\to\mathbb R$. Thus, by picking a real number $ \lambda $ called cutpoint, we can split the strings into three different sets: those having accepting value (i) less than $ \lambda $, (ii) equal to $\lambda$, and (iii) greater than $ \lambda $. Each set or two of them form a language defined by $ (G,\lambda) $, i.e $ L(G,\diamond \lambda) $ for $ \diamond \in \Diamond $. More specifically, in the case that $\diamond$ equals to $>$, we have $L(G,>\lambda)= \{ w \in \Sigma^* \mid f_G(w) > \lambda \}$, etc.

Moreover, notation  $ L(G,\lambda) $ refers to the triple $ ( L(G, < \lambda),L(G, = \lambda),L(G, > \lambda) ) $ and $ L^+(G,\lambda) $ refers to the same triple as except that $ \varepsilon $ is removed. 
Remark that for any given pair $ (G,\lambda) $ and another cutpoint $ \lambda' $, there is always another GFA $ G' $ (that can be designed based on $ G $ by adding an additional state) such that
\[
	L(G,\lambda) = L(G',\lambda'),
\]
Therefore, the choice of the cutpoint is not essential unless the number of states is significant.

Let $ G $ be a given GFA, $ \lambda $ be a cutpoint, and $ \diamond \in \Diamond  $, then determining whether $ L(G,\diamond \lambda) = \emptyset $  is an {\em emptiness} problem. Remark that the description of $ G $ and $ \lambda $ should be finite in order to be a ``reasonable'' computational problem. So, we can formulate the emptiness problem for computable numbers in general and for rational (or integer) numbers in a restricted case. The emptiness problems can be formulated for all the automata models in this paper in a straightforward way, but the special cases are emphasized.

Probabilistic finite automaton (PFA) is a special case of GFA such that only stochastic transition matrices and vectors are allowed. Moreover, a PFA gives its decision based on a subset of states called accepting states. Formally,  a PFA $P$ is 5-tuple
\[
	P = (Q,\TSigma,\{ A_\sigma \mid \sigma \in \TSigma \},v_0,Q_a),
\]
where, different from a GFA,
\begin{itemize}
\item $ A_\sigma $ is a (left) stochastic transition matrix for symbol $ \sigma $, i.e. $ A[j.i] $ represents the probability of $ P $ going from $ i $-th state to $ j $-th state after reading symbol $ \sigma $,
\item $ v_0 $ is a stochastic vector called initial probabilistic state, and,
\item $ Q_a \subseteq Q $ is called the set of accepting states.
\end{itemize}
For a given input $ w \in \Sigma^* $, $ P $ starts the computation in $ v_0 $ and reads $ w\dollar $ symbol by symbol in the same way of a GFA:
\[
	v_j = A_{\tildew_{j-1}} v_{j-1},
\]
where $ 1\leq j \leq |\tildew|$. The final state is $v_f = v_{|\tildew|} $. The accepting probability of $ w $ is calculated as
\[
	f_P(w) = \sum_{q_i \in Q_a} v_f[i].
\]
The languages recognized by $ P $ are defined in a similar way but now $ \lambda \in [0,1] $. Any language defined as $ L(P,> \lambda) $ (resp. $ L(P,\neq \lambda) $) is called stochastic (resp. exclusive stochastic) \cite{Rab63,Paz71}. Moreover, a nondeterministic finite automaton (NFA) can be defined as a PFA $ P $ of form $ L(P,>0) $. Similarly, a universal finite automaton (UFA) can be defined as a PFA $ P $ of form $  L(P,=1)$. So, for any PFA $P$, the language $ L(P,>0) $ or $ L(P,=1) $ (or $ L(P,<1) $ or $ L(P,=0) $) must be regular. On the other hand, any $ L(P,\diamond \lambda) $ for some $ \lambda \in (0,1) $ and $ \diamond \in \Diamond $ does not need to be regular.

If each transition matrix of a PFA is double stochastic (both column and row summations are 1), it is called {\em bistochastic} PFA (BPFA) (see \cite{Tur75}).

A quantum finite automaton (QFA) is a non-trivial generalization of PFA that can be in mixture of quantum states (mixed states) and use superoperators for transitions \cite{Hir11,YS11A}. Formally, a QFA $ M $ is a 5-tuple
\[
	M = (Q,\Sigma,\{ \mathcal{E}_\sigma \mid \sigma \in \TSigma \},\rho_0,Q_a),
\]
where, different from PFA, $ \mathcal{E}_\sigma = \{ E_{\sigma,j} \mid 1 \leq j \leq l_\sigma \} $ is a superoperator with $ l_\sigma $ operation elements that is applied when reading symbol $ \sigma $ and $ \rho_0 $ is the initial mixed state.
For a given input $ w \in \Sigma^* $, $ M $ starts the computation in $ \rho_0 $, reads $ w\dollar $ symbol by symbol in the same way of a PFA:
\[
	\rho_j = \mathcal{E}_{\tildew_j} (\rho_{j-1}) = \sum_{k=1}^{l_\sigma} E_{\sigma,k} \rho_{j-1} E_{\sigma,k}^\dagger ,
\]
where $ 1\leq j \leq |\tildew|$. The final state is $\rho_f = \rho_{|\tildew|} $. The accepting probability of $ w $ is calculated as
\[
	f_M(w) = \sum_{q_i \in Q_a} Tr(\rho_f[i,i]).
\]

Similar to PFAs, a nondeterministic quantum finite automaton (NQFA) is a QFA $ M $ that can define the single language $ L(M,>0) $ and a universal quantum finite automaton (UQFA) is a QFA $ M $ that can define the single language $ L(M,=1) $. On contrary to NFAs and UFAs, NQFAs and UQFAs can define nonregular languages \cite{YS10A}.

\subsection{Linear recurrence relations}

A linear recurrence relation (LRR) $u$ with depth $k>0$ (defined over real/complex numbers) is a double (initial values and coefficients)
\[
	u=( ( u_0,u_1,\ldots,u_{k-1} ), ( a_1,a_1,\ldots,a_{k} ))
\]
that defines an infinite sequence
\[
	u_0,u_1,\ldots,u_n,\ldots,
\]
where $ u_n = a_1 u_{n-1} + \cdots + a_k u_{n-k} $ for $ n \geq k $.

\subsection{Skolem's problem}
The Skolem's problem (known also as the Pisot problem, see \cite{HaHaHiKa} for a general introduction) is to determine whether a given linear recurrence relation
\[
u_n=a_1u_{n-1}+a_2u_{n-2}+\cdots+a_ku_{n-k}
\]
over integers with initial values $u_0$, $u_1$, $u_2$, $\ldots$, $u_{k-1}$, has a member in the sequence with value of 0 ($u_n=0$). Number $k$ in the definition is referred as to the {\em depth} of the recursion.

Currently it is not known if the Skolem's problem is decidable. The decidability is known for recurrence depths less than $5$, and allegedly for recursion depth $5$ \cite{HaHaHiKa}.

{\em The positivity problem} is to decide, for a linear recurrent sequence $u_n$, whether
$u_n\ge 0$ for each $n\ge 1$. {\em The strict positivity problem} is to decide whether $u_n>0$ for each $n\ge 1$.

\section{Linear recurrence automata}
\label{sec:LRA}

In this section, we introduce automata versions of LRRs and also some of their possible generalizations.

\subsection{Unary versions}

A linear recurrence automaton (LRA) $U = (u,\Sigma)$ with depth $k>0$ is a unary automaton composed by two elements: the LRR $u$ and the unary alphabet $ \Sigma $.  For each string $ a^j $, $U$ assigns $u_j$ as the accepting value:
\[
	f_{U}(a^j) = u_j, \mspace{20mu} j \geq 0.
\]

Consider the following LRA with depth $k>0$: 
\[
	U= (\{a\},u=(1,0,\ldots,0),(0,\ldots,0,1)).
\]
It is clear that $ u_n = u_{n-k} $ and the sequence contains only zeros except 
\[
	u_0=u_k=u_{2k}= \cdots = 1.
\] 
Therefore, all the following languages are the same and identical to $ \mathtt{MOD_k} = \{ a^j \mid j \mod k \equiv 0 \} $:
\[
	L(U,=1), L(U,\geq 1), L(U,>0), L(U,\neq 0). 
\]

A LRR can also be defined over mathematical objects other than numbers. A LRR can be generalized by defining over $ \mathbb{R}^m $ (or $ \mathbb{C}^n $) ($m>0$), called $m$-dimensional LRR over vectors (LRR-V) with depth $ k>0 $. Formally LRR-V $ V $ is a double
\[
	v = ( (v_0,v_1,\ldots,v_{k-1}), (A_1,\ldots,A_k)\})
\]
that defines an infinite sequence of vectors
\[
	v_0,v_1,\ldots,v_m,\ldots,
\]
where $ v_n = A_1 v_{n-1}+\cdots+A_k v_{n-k} $ for $ n \geq k $.

The automaton version of an m-dimensional LRR-V with depth $k$ (LRVA) is a triple $ V = (v,\Sigma,f) $, where $ v $ is a LRR-V and $ f \in \mathbb{R}^m $ is a row vector. For a given string $a^j$ over unary alphabet $\Sigma = \{a\}$, its accepting value is calculated as
\[
	f_V(a^j) = f \cdot v_j.
\]
The reader can notice that any $ m $-dimensional LRVA with depth 1 is an $m$-state unary GFA. Thus LRVAs can be seen as generalizations of unary GFAs.

The probabilistic version of a LRVA (PLRVA) is a restricted LRVA such that
\begin{enumerate}
	\item All elements in the sequence including the initial ones are stochastic vectors.
    \item Each matrix in $ \{ A_1,\ldots,A_k \} $ is a non-negative multiple of a stochastic matrix providing that the summation $ A_1 + \cdots + A_k $ gives a stochastic matrix, i.e. $ A_i = d_i B_i $ ($ d_i \geq 0 $) for some stochastic matrix $ B_i $ and $ d_1 + \cdots + d_m = 1 $.
    \item Each entry of $ f $ is in $ [0,1] $.
\end{enumerate}
Thus, the accepting values of each string is in $ [0,1] $, which can also be called accepting probability. It is clear that any PLRVA with depth 1 is a unary PFA.

We leave as a future work the definitions of complex and possible quantum versions of LRVAs.

\subsection{Binary versions}
Since LRRs can be seen as unary automata, it is natural to think the binary (or $n$-ary) versions of LRRs and their automata versions.

An LRR defines an infinite sequence. In the binary case, we define an infinite binary tree where each node represents a number (or an object) and so each top to bottom path is an infinite sequence. In order to specify such a sequence, we label each node as a string defined over $ \Sigma = \{a,b\} $. The root is empty string. The first level nodes are labeled from left to right as $ a $ and $ b $. The second level nodes are labeled from left to right as $ aa $, $ ab $, $ ba $, and $ bb $. In general, if a node is labeled with string $ w $, its left child is labeled with $ wa $ and its right child is labeled with $ wb $. 

A tree linear recurrence relation (T-LRR) $ t $ with depth $ k $ is a triple:
\[
	t = (\Sigma, (t_w \mid w \in \Sigma^{\leq k-1}),(a_1,\ldots,a_k,b_1,\ldots,b_k) ),
\]
where $ \Sigma $ is a finite alphabet, $ \Sigma^{\leq k-1} $ is set of all strings with lengths less than $ k $, $ t_w $ is the initial value for the node/string $ w $ for $ w \in \Sigma^{\leq k-1} $, and $ (a_1,\ldots,a_k,b_1,\ldots,b_k)  $ are coefficients. For any given $ w=w_1 w_2 \cdots w_{n-k} \cdots w_{n-1} w_{n} \in \Sigma^* $ with length $ n \geq k $, $ t_w $ can be calculated as
\[
	t_w = c_1 t_{w_{n-1}} + c_2 t_{w_{n-2}} + \cdots + c_k t_{w_{n-k}},
\]
where $  c_i $ is $ a_i $ if $ w_{n-i} $ is symbol $ a $ and $  c_i $ is $ b_i $ if $ w_{n-i} $ is symbol $ b $.

Any T-LRR can also be seen as an automaton (T-LRA) since its definition includes the alphabet. With depth 1, we obtain a 1-state GFA, for example. If $ |\Sigma| = 1 $, then we obtain a LRR. The definitions of T-LRR and T-LRA can be extended with some other mathematical objects like vectors, i.e. T-LRR-V and T-LRVA, respectively. Then, we can obtain generalizations of GFAs. In case of probabilistic T-LRA over vectors (T-PLRVA), we can assume that each matrix is some non-negative multiple of a stochastic matrix and each vector in the tree is normalized to be stochastic. We can also assume to fix each multiple coefficient to $ \frac{1}{k} $ for the simplicity, where $ k $ is the depth of recursion.

\section{Basic Facts}
\label{sec:basic}

Linearly recurrent sequences can be characterized in multiple ways. We give their direct relations with GFAs, and so with PFAs and QFAs.
 
\begin{proposition}\label{FormulationLemma} For an integer sequence $u_0$, $u_1$, $u_2$, $\ldots$,
the following are equivalent:
\begin{enumerate}
\item The sequence $u_n$ is a linear recurrent sequence.
\item For $n\ge 1$, $u_n=M^n[k,1]$, where $M\in\mathbb Z^{k\times k}$ for some $k$.
\item For $n\ge 1$, $u_n= x M^n y $, where
$ x \in \mathbb{Z}^k $ is a row vector,  $ y \in\mathbb Z^k$ is a column vector, and $M\in\mathbb Z^{k\times k}$ for some $k$.
\end{enumerate}
\end{proposition}

\begin{proof} Implication (1) $\implies (2)$: 
Assume that a sequence $u_n$ is given by first fixing $u_0$, $\ldots$, $u_{k-1}$, and for $n\ge k$ defined by
recurrence
\[
u_{n}=a_{k-1}u_{n-1}+\cdots+a_1u_{n-k+1}+a_0u_{n-k}.
\]
We define
\begin{equation}
\label{RecMatrix}
M_1=\left(\begin{array}{c|c}
a_{k-1} \cdots a_2 ~ a_1 & a_0 
\\ \hline
 & 0 \\
 I & \vdots \\
 & 0 \\
\end{array}\right)
\end{equation}
It is easy to see that for each $n\ge 0$,
\[
u_n= x M_1^n y,
\]
where $ x=(u_{k-1} ~\cdots~u_1~u_0)$, and $y=(0~\cdots~0~1)^T$. We denote $\vec 0=(0~~0~~\cdots~~0)$ and define a $(k+1)\times (k+1)$-matrix $M$ by
\[
M=\left(\begin{array}{cc}
0 & \vec 0^T \\
M_1 y & M_1
\end{array}\right)
=
\mymatrix{c|c}{ 0 & \vec 0 \\ \hline u_k & \\ u_{k-1} & \\ \vdots & ~~~~M_1~~~~ \\ u_2 & \\ u_1 & }
\]
Inductively we see that
\[
M^n=\left(\begin{array}{cc}
0 &\vec 0 \\
M_1^n x & M_1^n
\end{array}\right),
\]
and, furthermore, that
\[
(M^n)[k+1,1]=
(0~~y)
\left(\begin{array}{cc}
0 &\vec 0 \\
M_1^n x & M_1^n
\end{array}\right)
\myvector{1\\ \vec 0}
=
x M_1^n y = u_n
\]
whenever $n\ge 1$.

Implication (2) $\implies$ (3) follows directly from equation
\[
M^n[k+1,1] = (0~~y) M^n \myvector{1 \\ \vec 0} = u_n.
\]

Implication (3) $\implies$ (1): 
Let  $p(x)=x^k-a_{k-1}x^{k-1}-\cdots -a_1x-a_0$ be the
characteristic polynomial of matrix $M$. According to the Cayley-Hamilton theorem \cite{Cohn}
\[
M^k=a_{k-1}M^{k-1}+\ldots+a_1M+a_0I,
\]
and consequently
\[
	M^{n}=a_{k-1}M^{n-1}+\cdots+a_1M^{n-k+1}+a_0M^{n-k}
\]
for any $n\ge k$. It follows that
\[
	x M^{n} y =a_{k-1} x M^{n-1} y +\cdots+a_1 x M^{n-k+1} y +a_0
x M^{n-k} y,
\]
which is to say that
\begin{equation}\label{rec1}
u_{n}=a_{k-1}u_{n-1}+\cdots+a_1u_{n-k+1}+a_0u_{n-k},
\end{equation}
so (\ref{rec1}) is the desired recurrence.
\qed\end{proof}

\begin{corollary}
	Let $ u_n $ be a linear integer recurrent sequence  with depth $k>0$. Then, there exists a $(k+1)$-state unary integer GFA $G$ such that 
    \[
    	u_n = f_G(a^{n})
    \]
    for $n \geq 1$.
\end{corollary}

\begin{corollary}
	Let $ G $ be a $k$-state unary integer GFA. Then, there exists a linear (possibly not integer) recurrent sequence $u_n$ with depth $k$ such that
    \[
    	u_n = f_G(a^{n})
    \]
    for $n \geq 0$, where $k>0$.
\end{corollary}

\begin{proposition}\label{CSP} Any constant sequence is a linear recurrent sequence. Moreover, if $u_n$ and $v_n$ are linear recurrent sequences, so are $u_nv_n$ and $u_n+v_n$.
\end{proposition}

\begin{proof} The claim for constant sequences is trivial. Using Proposition \ref{FormulationLemma}, we can write $u_n= x_1 M_1^n y_1 $, and $v_n= x_2 M_2^n y_2$. Now $u_n v_n=(x_1 \otimes x_2)(M_1\otimes M_2)^n(y_1\otimes y_2)$ (tensor product construction)
and $u_n+v_n=(x_1 \oplus x_2)(M_1\oplus M_2)^n(y_1\oplus y_2)$ (direct sum construction), and Proposition \ref{FormulationLemma} implies that these sequences are linear recurrent. 
\qed\end{proof}

\begin{proposition} 
	\label{pro:LRR-reductions}
	Skolem's problem is reducible to the positivity problem and the strict positivity problem. The strict positivity problem is reducible to the positivity problem.
\end{proposition}
\begin{proof}  For the first reduction: Given a linear recurrent sequence $u_n$, we can use the constructions of Proposition \ref{CSP} to create another linear recurrent sequence $ v_n = -u_n^2 +1 $. It is then clear that $ v_n > 0 $ if and only if $ u_n = 0 $.

For the second reduction: Similarly, we create $ v_n = -u_n^2$. It is then clear that $v_n \geq 0$ if and only if $u_n = 0$. 

For the third reduction: In the same way, we create $ v_n = 2u_n +1 $. It is then clear that $v_n > 0$ if and only if $u_n \geq 0$. 
\qed\end{proof}

\begin{remark}\label{Rem02}
Positivity problem is known to be decidable for linear recurrences with depth at most $5$ \cite{OW14A}. In addition, the positivity problem for {\em simple} recurrences (those having no multiple roots) is known to be decidable for recursion depth at most $9$ \cite{OW14B}. The {\em ultimate positivity problem} is decidable for simple linear recurrences of any order \cite{OW14C}.
\end{remark}

\section{Conversions between unary automata models}
\label{sec:conversions}

Conversions between unary automata models can be summarized as follows (see \cite{Tur75,Hir11,YS10A,YS11A,SayY14} for the conversion between GFAs, PFAs, and QFAs):
\begin{itemize}
	\item For any given integer LRA $ U = (u,\Sigma) $ with depth $ k $, there exists a $(k+1)$-state integer GFA $ G $ such that 
	\[
		L^+(U,0) = L^+(G,0).
	\]
	\item For any given rational $n$-state GFA $ G $ and cutpoint $ \lambda \neq 0 $, there exist a rational $ (n+1) $-state rational $ G' $ such that
	\[
		L(G,\lambda) = L(G',0)
	\]
	Moreover, there exists an integer $ (n+1) $-state GFA $ G'' $ such that
	\[
		L(G,\lambda) = L(G',0) = L(G'',0).
	\]
	\item For any given $ n $-state rational GFA $ G $, there exist $ (n+2) $-state rational PFA $ P $ and $ (n+2) $-state algebraic number QFA $M$ such that
	\[
		L(G,0) = L(P,\frac{1}{n+2}) = L(M,\frac{1}{n+2}).
	\]
	\item For any given $ n $-state rational QFA $ M $, there exists an $ n^2 $-state rational GFA $ G $ such that
	\[
		L(M,\lambda) = L(G,\lambda) \mbox{ for any } \lambda \in [0,1].
	\]
	\item For any given QFA $ M $, there exists a QFA $ M' $ such that
	\[
		L(M,\neq \lambda) = L(M',\neq \lambda' ) \mbox{ for any } \lambda, \lambda' \in [0,1].
	\]
	\item For any given $ n $-state integer GFA $ G $, there exists an integer LRA $ U $ with depth $ n $ such that
	\[
		L^+(G,\lambda) = L^+(U,\lambda) \mbox{ for any } \lambda \in \mathbb{R}.
	\]
\end{itemize}

\section{The status of emptiness problems}
\label{sec:status}

In this section, we list the decidable problems and the problems which are still open with their restricted but decidable cases.   

\subsection{Problems}

The decision problems for LRRs can be formulated by LRAs as follows:
\begin{itemize}
\item Skolem's problem: 
	\begin{itemize}
		\item For a given integer LRA $ U $, is $ L(U,=0) $ empty? 
		\item More generally, is $ L(U,=\lambda) $ empty for a given $ \lambda \in \mathbb{R} $?
	\end{itemize}	 
\item The positivity problem: 
	\begin{itemize}
		\item For a given integer LRA $ U $, is $ L(U,\geq 0) $ empty? 
		\item More generally, is $ L(U,\geq \lambda) $ empty for a given $ \lambda \in \mathbb{R} $?
		\item It is also equivalent to this decision problem: is $ L(U,\leq \lambda) $ empty for a given $ \lambda \in \mathbb{R} $? 
	\end{itemize} 
\item The strict positivity problem: 
	\begin{itemize}
		\item For a given integer LRA $ U $, is $ L(U,> 0) $ empty? 
		\item More generally, is $ L(U,> \lambda) $ empty for a given $ \lambda \in \mathbb{R} $?
		\item It is also equivalent to this decision problem: is $ L(U,< \lambda) $ empty for a given $ \lambda \in \mathbb{R} $? 
	\end{itemize} 
\item The \textit{exclusivity}\footnote{We propose this term inspired by the definition of exclusive stochastic languages.} problem:
	\begin{itemize}	
	 	\item For a given integer LRA $ U $, is $ L(U, \neq 0) $ empty? 
		\item More generally, is $ L(U, \neq \lambda) $ empty for a given $ \lambda \in \mathbb{R} $?
	\end{itemize}
\end{itemize}

Due to Proposition \ref{pro:LRR-reductions}, we know that Skolem's problem can be reduced to the strict positivity problem that can be reduced to  positivity problem. Thus the most general problem is the positivity problem. 

Due to the conversions given in the previous section, each decision problem given for LRRs and so LRAs has an equivalent emptiness problem for GFAs, PFAs, and QFAs. In the above list, it is sufficient to change LRA $ U $ with unary GFA, PFA, or QFA model in order to obtain the corresponding emptiness problem. 

However, remark that the value of $ \lambda $ can be only in $ [0,1] $ for PFAs and QFAs (the cases for the other $\lambda$ values are trivial) and we should also note that their behaviors are different on the borders.

\subsection{Decidable problems}

PFAs define only regular languages when $ \lambda $ is 0 or 1. So, any emptiness problem in such case is decidable.

The emptiness problem for NQFAs defined with algebraic numbers is known to be decidable \cite{DHRSY14}. Therefore, the problem of whether $ L(M,>0) $ or $ L(M,<1) $ is empty set for a given QFA $ M $ (defined with algebraic numbers) is decidable. Remark that the problem of whether $ L(M,\geq 0) $ or $ L(M,\leq 1) $ is empty for a given QFA $M$ is trivial. 

The remaining cases are to determine whether  $ L(M, = 0) $ or $ L(M, = 1) $ is empty for a given QFA $ M $ (the emptiness problem for UQFAs) are still open. 

For a given QFA $ M $, $ L(M,>0) $ is the same as $ L(M,\neq 0) $, and so we can also follow that the problem of determining whether $ L(M, \neq \lambda) $ is empty set is also decidable for any $ \lambda \in [0,1] $. Thus, the exclusivity problem is decidable for LRAs, GFAs, PFAs, and QFAs.

\subsection{Open problems}

Now, we can list the open cases:
\begin{itemize}
	\item Skolem's problem: LRAs, GFAs, PFAs when $ \lambda \in (0,1) $, and QFAs when $ \lambda \in [0,1] $. 
	\item The positivity problem: LRAs, GFAs, PFAs when $ \lambda \in (0,1) $, and QFAs when $ \lambda \in (0,1) $.
	\item The strict positivity problem: LRAs, GFA, PFAs when $ \lambda \in (0,1) $, and QFAs when $ \lambda \in (0,1) $.
\end{itemize}
We do not know the answer of any question for any models defined with computable number or any subsets of computable numbers. We know only some decidable results for small automata given in the next subsection.

\subsection{Special cases}

From Remark \ref{Rem02}, we know that the positivity problem is decidable for depth 5. Thus the same problems are decidable also for the following models by using the conversions above:
\begin{itemize}
\item 5-state integer GFAs,
\item 4-state rational GFAs and also rational PFAs, and,
\item 2-state rational QFAs. 
\end{itemize}


\bibliographystyle{plain}
\bibliography{tcs}

\begin{thebibliography}{10}

\bibitem{AY15}
Andris Ambainis and Abuzer Yakary{\i}lmaz.
\newblock Automata and quantum computing.
\newblock Technical Report 1507.01988, arXiv, 2015.

\bibitem{Cohn}
P.~M. Cohn.
\newblock {\em Algebra}, volume~2.
\newblock John Wiley and Sons, 1977.

\bibitem{DHRSY14}
H.~G{\"{o}}kalp Demirci, Mika Hirvensalo, Klaus Reinhardt, A.~C.~Cem Say, and
  Abuzer Yakary{\i}lmaz.
\newblock Classical and quantum realtime alternating automata.
\newblock In {\em NCMA}, volume 304, pages 101--114. {\"{O}}sterreichische
  Computer Gesellschaft, 2014.
\newblock (arXiv:1407.0334).

\bibitem{Ei74}
Samuel Eilenberg.
\newblock {\em Automata, Languages, and Machines Volume A}.
\newblock Academic Press, London, 1974.

\bibitem{HaHaHiKa}
Vesa Halava, Tero Harju, Mika Hirvensalo, and Juhani Karhum\"{a}ki.
\newblock Skolem's problem -- on the border between decidability and
  undecidability.
\newblock TUCS Technical Reports 683, Turku Centre for Computer Science, 2005.

\bibitem{Hir11}
Mika Hirvensalo.
\newblock Quantum automata theory -- a review.
\newblock In {\em Algebraic Foundations in Computer Science}, volume 7020 of
  {\em LNCS}, pages 146--167, 2011.

\bibitem{Kleene}
Stephen~C. Kleene.
\newblock Representation of events in nerve nets and finite automata.
\newblock In {\em Automata Studies}, pages 3--41. Princeton University Press,
  1956.

\bibitem{KW97}
A.~Kondacs and J.~Watrous.
\newblock On the power of quantum finite state automata.
\newblock In {\em FOCS'97}, pages 66--75, 1997.

\bibitem{McCulloch}
Warren McCulloch and Walter Pitts.
\newblock A logical calculus of the ideas immanent in nervous activity.
\newblock {\em Bulletin of Mathematical Biophysics}, 7:115--133, 1943.

\bibitem{Mealy}
George Mealy.
\newblock A method for synthesizing sequential circuits.
\newblock {\em Bell Systems Technical Journal}, 34:1045--1079, 1955.

\bibitem{MC00}
Cristopher Moore and James~P. Crutchfield.
\newblock Quantum automata and quantum grammars.
\newblock {\em Theoretical Computer Science}, 237(1-2):275--306, 2000.

\bibitem{Moore}
Edward Moore.
\newblock Gedanken-experiments on sequential machines.
\newblock In {\em Automata Studies}, pages 129--153. Princeton University
  Press, 1956.

\bibitem{OW14B}
Jo{\"{e}}l Ouaknine and James Worrell.
\newblock On the positivity problem for simple linear recurrence sequences,.
\newblock In {\em Automata, Languages, and Programming - 41st International
  Colloquium, {ICALP} 2014, Part {II}}, volume 8573 of {\em Lecture Notes in
  Computer Science}, pages 318--329. Springer, 2014.

\bibitem{OW14A}
Jo{\"{e}}l Ouaknine and James Worrell.
\newblock Positivity problems for low-order linear recurrence sequences.
\newblock In {\em Proceedings of the Twenty-Fifth Annual {ACM-SIAM} Symposium
  on Discrete Algorithms ({SODA}'14)}, pages 366--379. {SIAM}, 2014.

\bibitem{OW14C}
Jo{\"{e}}l Ouaknine and James Worrell.
\newblock Ultimate positivity is decidable for simple linear recurrence
  sequences.
\newblock In {\em Automata, Languages, and Programming - 41st International
  Colloquium, {ICALP} 2014, Part {II}}, volume 8573 of {\em Lecture Notes in
  Computer Science}, pages 330--341. Springer, 2014.

\bibitem{Paz71}
A.~Paz.
\newblock {\em Introduction to Probabilistic Automata}.
\newblock Academic Press, New York, 1971.

\bibitem{Rab63}
M.~O. Rabin.
\newblock Probabilistic automata.
\newblock {\em Information and Control}, 6:230--243, 1963.

\bibitem{RS59}
M.O. Rabin and D.~Scott.
\newblock Finite automata and their decision problems.
\newblock {\em IBM Journal of Research and Development}, 3:114--125, 1959.

\bibitem{SayY14}
A.~C.~Cem Say and Abuzer Yakary{\i}lmaz.
\newblock Quantum finite automata: A modern introduction.
\newblock In {\em Computing with New Resources}, volume 8808 of {\em LNCS},
  pages 208--222. Springer, 2014.

\bibitem{Tur75}
Paavo Turakainenn.
\newblock Word-functions of stochastic and pseudo stochastic automata.
\newblock {\em Annales Academiae Scientiarum Fennicae, Series A. I,
  Mathematica}, 1:27--37, 1975.

\bibitem{YS10A}
Abuzer Yakary{\i}lmaz and A.~C.~Cem Say.
\newblock Languages recognized by nondeterministic quantum finite automata.
\newblock {\em Quantum Information and Computation}, 10(9\&10):747--770, 2010.

\bibitem{YS11A}
Abuzer Yakary{\i}lmaz and A.~C.~Cem Say.
\newblock Unbounded-error quantum computation with small space bounds.
\newblock {\em Information and Computation}, 279(6):873--892, 2011.

\bibitem{Yu}
Sheng Yu.
\newblock Regular languages.
\newblock In Grzegorz Rozenberg and Arto Salomaa, editors, {\em Handbook of
  Formal Languages, Vol. 1}, pages 41--110. Springer-Verlag New York, Inc.,
  1997.

\end{thebibliography}

\end{document}